\documentclass[envcountsame,a4paper]{llncs}

\usepackage{latexsym}
\usepackage{amssymb}
\usepackage{epsfig}
\usepackage{amsfonts}
\usepackage{color}
\usepackage{amsmath}
\usepackage{fleqn}
\usepackage[T1]{fontenc}
\usepackage{ifthen}
\usepackage{tkz-graph}

\newboolean{commentson} 
\setboolean{commentson}{true}

\newcommand{\comment}[1]
{\ifthenelse{\boolean{commentson}}
   {{\par\noindent\mbox{}{\small[ *** #1 ]\par}\noindent\par}}{}}

\newcommand{\VCSP}{\mbox{\rm VCSP}}
\newcommand{\MinHOM}{\mbox{\rm Min HOM}}

\newcommand{\wMCSP}{\mbox{\sc  Min CSP}}

\newcommand{\NP}{\mbox{\bf NP}}

\def\2mat#1#2#3#4#5#6#7#8{
\begin{array}{c|cc}
$~$ & #3 & #4\\
\hline
#1 & #5& #6\\
#2 & #7 & #8 \end{array}}

\renewcommand{\phi}{\varphi}
\renewcommand{\epsilon}{\varepsilon}

\renewcommand{\text}[1]{\mbox{\rm \,#1\,}}        

\renewcommand{\emptyset}{\varnothing}  

\newcommand{\ie}{{\em i.e.}}                
\newcommand{\cf}{{\em cf.}}
\newcommand{\eg}{{\em e.g.}}

\newcommand{\unprint}[1]{}

\newcommand{\Z}{\mbox{$\mathbb Z$}}

\newcommand{\Qnn}{\mbox{$\mathbb Q_{\geq 0}$}}

\renewcommand{\rho}{\varrho}

\newcommand{\Endo}[1]{\text{End}(\ensuremath{#1})}

\renewcommand{\emptyset}{\varnothing}  

\newcommand{\meet}{\wedge}
\newcommand{\join}{\vee}

\newcommand{\etal}{{et al.\ }}

\newcommand{\optsol}{\mbox{\rm Optsol}}

\newcommand{\bfsym}[1]{\boldsymbol{#1}}
\newcommand{\consv}[1]{\overrightarrow{#1}}

\newcommand{\ST}[0]{V'}

\newcommand{\po}[0]{<}
\newcommand{\lub}[0]{\join}
\newcommand{\glb}[0]{\meet}

\begin{document}

\usetikzlibrary{calc}
\tikzstyle{vertex8}=[draw,circle,fill=black,text=white,minimum size=4pt,inner sep=0pt]
\tikzstyle{edge}=[]
\tikzstyle{edged}=[dotted]

\pagestyle{plain}

\title{Min CSP on Four Elements: Moving Beyond Submodularity}
\author{Peter Jonsson\inst{1}\thanks{Partially supported by the {\em Swedish
Research Council} (VR) under grant 621-2009-4431.} \and Fredrik Kuivinen\inst{2} \and Johan Thapper\inst{3}\thanks{Supported by the LIX-Qualcomm Postdoctoral Fellowship.}}
\institute{Department of Computer and Information Science,\\
Link\"{o}pings universitet,\\
SE-581 83 Link\"{o}ping, Sweden\\
\email{petej@ida.liu.se} \and 
\email{frekui@gmail.com} \and 
\'{E}cole polytechnique,
Laboratoire d'informatique (LIX),\\
91128 Palaiseau Cedex, France\\
\email{thapper@lix.polytechnique.fr}}

\maketitle

\begin{abstract}
We report new results on the complexity of the valued
constraint satisfaction problem (VCSP).
Under the unique games conjecture, 
the approximability
of finite-valued VCSP is fairly well-understood.
However, there is yet no characterisation of VCSPs that can be solved
exactly in polynomial time. This is unsatisfactory, since
such results are interesting from a
combinatorial optimisation perspective; there are deep connections
with, for instance, submodular and bisubmodular minimisation.
We consider the Min and Max CSP problems (\ie{} where the cost functions
only attain values in $\{0,1\}$) over four-element domains
and identify all tractable fragments.
Similar classifications were previously known for two- and three-element
domains.
 In the process, we introduce
a new class of tractable VCSPs based on a
generalisation of submodularity.
We also extend and modify a graph-based technique
by Kolmogorov and \v{Z}ivn{\'y} (originally introduced by Takhanov) 
for efficiently obtaining hardness results in our setting.
This allow us to prove the result without relying on computer-assisted
case analyses (which otherwise are fairly common when studying the
complexity and approximability of VCSPs.) The hardness results
are further simplified by the introduction of powerful reduction
techniques.
\medskip

\end{abstract}

\section{Introduction}

This paper concerns the computational complexity of an optimisation problem with
strong connections to the \emph{constraint satisfaction problem} (CSP). 
An instance of the constraint satisfaction problem consists of a finite set of
variables, a set of values (the domain), and a finite set of constraints.
The goal is to determine whether there is an assignment of values to the variables
such that all the constraints are satisfied.
CSPs provide a general framework for modelling a variety of combinatorial decision
problems~\cite{Creignou:Constraints,Hell:Nesetril:CCC}.

Various optimisation variations of the constraint satisfaction framework have been
proposed and
many of them can be seen as special cases of the valued constraint satisfaction
problem (\VCSP), introduced by Schiex \etal\cite{Schiex:IJCAI95}.
This is an optimisation problem which is general enough to express such diverse
problems as {\sc Max CSP}, where the goal is to maximise the number of satisfied
constraints, and the minimum cost homomorphism problem (\MinHOM), where all
constraints must be satisfied, but each variable-value tuple in the assignment is
given an independent cost.
To accomplish this, instances of the \VCSP{} assign costs (possibly infinite) to
individual tuples of the constraints. It is then convenient to replace relations by
\emph{cost functions}, \ie{} functions from tuples of the domain to some set of
costs.
This set of costs can be relatively general, but much is captured by using $\Qnn \cup \{ \infty \}$, where \Qnn{} denotes the set of non-negative rational numbers.
We arrive at the following formal definition.

\begin{definition}
Let $D$ be a finite domain, and let $\Gamma$ be a set of functions $f_i : D^{k_i}
\rightarrow \Qnn \cup \{ \infty \}$.
By VCSP($\Gamma$) we denote the following minimisation problem:

\begin{description}
\item[Instance:] A set of variables $V$, and a sum $\sum_{i=1}^m \rho_i
f_i(\bfsym{x_i})$,
where $\rho_i \in \Qnn$, $f_i \in \Gamma$, and $\bfsym{x_i}$ is a list of $k_i$
variables from $V$.
\item[Solution:] A function $\sigma : V \rightarrow D$.
\item[Measure:]
$m(\sigma) = \sum_{i=1}^m \rho_i f_i(\sigma(\bfsym{x_i}))$,
where $\sigma(\bfsym{x_i})$ is the list of elements from $D$ obtained 
by applying $\sigma$ component-wise to $\bfsym{x_i}$.
\end{description}
\end{definition}

The set $\Gamma$ is often referred to as the {\em constraint language}.
We will use $\Gamma$ as our parameter throughout the paper. For instance,
when we say that a class of VCSPs $X$ is polynomial-time
solvable, then we mean that VCSP$(\Gamma)$ is polynomial-time
solvable for every $\Gamma \in X$.
Finite-valued functions, \ie{} functions with a range in \Qnn, are sometimes called
\emph{soft constraints}. A prominent example is given by
functions with a range in $\{0,1\}$; they can be used to express instances of
the well-known 
{\sc Min CSP} and {\sc Max CSP} problems (which, for instance, include 
{\sc Max} $k$-{\sc Cut}, {\sc Max} $k$-{\sc Sat},
and {\sc Nearest Codeword} as subproblems).
On the other side we have \emph{crisp constraints} which represent the standard type
of CSP constraints.
These can be expressed by cost functions taking values in $\{0,\infty\}$.

A systematic study of the computational complexity of the \VCSP{} was initiated
by Cohen \etal\cite{Cohen:etal:AI06}; for instance,
they prove a complexity dichotomy for \VCSP$(\Gamma)$ over 
two-element domains. This was the starting point for an 
intensive research effort leading to a 
large number of complexity results
for VCSP: examples include complete classifications of
conservative constraint languages (\ie{} languages containing all unary cost functions)~\cite{Deineko:etal:JACM2008,Zivny:1,Kolmogorov:1},
$\{0,1\}$ languages on three elements~\cite{Jonsson:etal:sicomp06},
languages containing a single $\{0,1\}$ cost function~\cite{Jonsson:Krokhin:JCSS07},
and arbitrary languages with $\{0,\infty\}$ cost
functions~\cite{Takhanov:STACS2010}. We note that some of these results have been
proved by
computer-assisted search---something that drastically reduces the
readability, and insight gained from the proofs.
We also note that there is no generally 
accepted conjecture stating which
VCSPs are polynomial-time solvable.

The picture is clearer when considering the approximability of finite-valued
VCSP. Raghavendra~\cite{Raghavendra:08} have presented algorithms for
approximating any finite-valued
VCSP. These algorithms achieve an optimal approximation ratio for the
constraint languages that cannot be solved to optimality in polynomial
time, given
that the unique games conjecture (UGC) is true. For the constraint
languages that can be solved to optimality,
one gets a PTAS from these algorithms. Furthermore, no
characterisation of the set of constraint languages that can be solved to
optimality
follows from Raghavendra's result. Thus, Raghavendra's result does
not imply the complexity results discussed above (not even conditionally under
the UGC).


The goal of this paper is to 
prove a dichotomy result for VCSP with $\{0,1\}$ cost
functions over four-element
domains: we show that every such problem is either solvable in polynomial time
or NP-hard. 
Such a dichotomy result is not known for CSPs on four-element domains
(and, consequently, not for unrestricted VCSPs on four-element domains).
Our result proves that, in contrast to the
two-element, three-element, and conservative case, submodularity is not
the only source of tractability.
In order to outline the proof, let $\Gamma$ denote a
constraint language with $\{0,1\}$ cost functions over a four-element
domain $D$.
We will need two tractability results in our classification.
The first one is well-known: if every function in $\Gamma$
is submodular on a chain (\ie{} a total ordering
of $D$), then VCSP$(\Gamma)$ is solvable in
polynomial time. The second result is new and can be found in
Section~\ref{sec:tractable}: we introduce
{\em 1-defect chain multimorphisms} and prove that
if $\Gamma$ has such a multimorphism, then VCSP$(\Gamma)$
is tractable. A multimorphism is, loosely speaking,
a pair of functions such that
$\Gamma$ satisfies certain invariance properties under them.
The algorithm we present is based on a combination of submodular and 
bisubmodular minimisation~\cite{IwataFF01,McCormickF10,Schrijver00}.

The hardness part of the proof consists of four parts
(Sections \ref{sec:endo}--\ref{sec:full}).
We begin by introducing some tools in Section~\ref{sec:endo} and~\ref{sec:graph}.
Section~\ref{sec:endo} concerns the problem of adding (crisp) constant unary
relations to $\Gamma$ without changing the computational complexity of the
resulting problem.
The main tool for doing this is using the concept of {\em indicator problems}
introduced by Jeavons \etal\cite{Jeavons:Cohen:Gyssens:Constraints1999}
(see also Cohen \etal\cite{Cohen:Cooper:Jeavons:CP06}).
Section~\ref{sec:graph} introduces a graph construction for studying $\Gamma$. 
In principle, this graph
provides information about the complexity of VCSP$(\Gamma)$ based on the
two-element sublanguages of $\Gamma$. Similar 
graphs has been used repeatedly in the study of
VCSP, \cf{}~\cite{Bulatov:LICS03,Zivny:1,Takhanov:STACS2010}. 
Equipped with these tools, we determine the complexity of VCSP$(\Gamma)$
over a four-element domain in Section~\ref{sec:binary}.
The graph introduced in Section~\ref{sec:graph} allows us to prove that,
when $\Gamma$ is a \emph{core} (\cf{} Section~\ref{sec:endo}),
VCSP$(\Gamma)$ is polynomial-time solvable if and only if $\Gamma$ is submodular
on a chain or $\Gamma$ has a 1-defect chain multimorphism
(Theorem~\ref{thm:fourclass}).
Some proofs of intermediate results are deferred to Section~\ref{sec:full}.

\section{Preliminaries}\label{sec:prel}

Throughout this paper, we will assume that $\Gamma$ is a finite set of
\{0,1\}-valued functions. 
By \wMCSP$(\Gamma)$ we denote the problem \VCSP$(\Gamma)$.
It turns out to be convenient to introduce a generalisation of this problem
in which we allow additional constraints on the solutions. From a \VCSP{} perspective, this means that we allow \emph{crisp} as well as $\{0,1\}$-valued cost functions.
To make the distinction clear, and since we will not be using any \emph{mixed} cost functions, we represent the crisp constraints with relations instead of $\{0,\infty\}$-valued cost functions.

\begin{definition}
  Let $\Gamma$ be a set of $\{0,1\}$-valued functions on a domain $D$,
  and let $\Delta$ be a set of finitary relations on $D$.
  By \wMCSP($\Gamma, \Delta$) we denote the following minimisation problem:
  \begin{description}
  \item[Instance:] A \wMCSP$(\Gamma)$-instance ${\cal I}$, and a finite set of constraint applications $\{ (\bfsym{y_j}; R_j) \}$, 
where $R_j \in \Delta$ and $\bfsym{y_j}$ is a matching list of variables from $V$.
\item[Solution:] A solution $\sigma$ to ${\cal I}$ such that $\sigma(\bfsym{y_j}) \in R_j$ for all $j$.
\item[Measure:] The measure of $\sigma$ as a solution to ${\cal I}$.
\end{description}
\end{definition}


 
We will generally omit the parenthesis surrounding singletons in unary relations, as in the following definition: let ${\cal C}_D = \{ \{d\} \mid d \in D \}$ be the set of constant unary relations over $D$.

\subsection{Expressive power and weighted pp-definitions}

It is often possible to enrich a set of functions $\Gamma$ without changing
the computational complexity of \wMCSP. 
In this paper, we will make use two distinct, but related notions
aimed at this purpose.

\begin{definition}
  \label{def:wpp}
Let $\mathcal{I}$ be an instance of \wMCSP$(\Gamma, \Delta)$, and let $\bfsym{x} = (x_1,\dots,x_s)$ be a sequence of distinct variables from $V({\cal I})$.
Let 
\[
\pi_{\bfsym{x}}\optsol({\cal I}) = \{ (\sigma(x_1),\dots,\sigma(x_s)) \mid \text{$\sigma$ is an optimal solution to ${\cal I}$} \},
\]
\ie{} the projection of the set of optimal solutions onto $\bfsym{x}$.
We say that such a relation \emph{has a weighted pp-definition in $(\Gamma, \Delta)$}.
  Let $\langle \Gamma, \Delta \rangle_w$ denote the set of relations which have a weighted pp-definition in $(\Gamma,\Delta)$.
\end{definition}

For an instance ${\cal J}$ of \wMCSP,
we define ${\sf Opt}({\cal J})$ to be the optimal value of
a solution to ${\cal J}$, and to be undefined if no solution exists.
The following definition is a variation of the concept of the
\emph{expressive power} of a valued constraint language,
see for example Cohen \etal\cite{Cohen:etal:AI06}.

\begin{definition}
  \label{def:express}
Let ${\cal I}$ be an instance of \wMCSP$(\Gamma, \Delta)$, and let
$\bfsym{x} = (x_1,\dots,x_k)$ be a sequence of distinct variables from $V({\cal I})$.
Define the function ${\cal I}_{\bfsym{x}} : D^k \rightarrow \Qnn$ by letting
${\cal I}_{\bfsym{x}}(a_1,\dots,a_k) = {\sf Opt}({\cal I} \cup \{ (x_i; \{a_i\}) \mid 1 \leq i \leq k \})$.
We say that ${\cal I}_{\bfsym{x}}$ is \emph{expressible over} $(\Gamma,\Delta)$.
  Let $\langle \Gamma, \Delta \rangle_{fn}$ denote the set of \emph{total functions} expressible over $(\Gamma,\Delta)$.
\end{definition}

\begin{proposition}
\label{prop:1}
  Let $\Gamma' \subseteq \langle \Gamma, \Delta \rangle_{fn}$ and
  $\Delta' \subseteq \langle \Gamma, \Delta \rangle_w$
  be finite sets.
  Then,
  \wMCSP$(\Gamma', \Delta')$ is polynomial-time reducible to
  \wMCSP$(\Gamma, \Delta)$.
\end{proposition}

\begin{proof}
  The reduction from \wMCSP$(\Gamma',\Delta')$ to \wMCSP$(\Gamma,\Delta')$ is a special case of Theorem 3.4 in~\cite{Cohen:etal:AI06}.
  We allow weights as a part of our instances, but this makes no essential difference.

  For the remaining part,
  we will assume that $\Delta' \setminus \Delta$ contains a single relation $R = \pi_{\bfsym{x}} \optsol({\cal J})$.
  The case when $\Delta' \setminus \Delta = \{R_1, \dots, R_k\}$, for $k > 1$ can be handled by eliminating one relation at a time using the same argument.
  Let $\mathcal{I'}$ be an instance of \wMCSP$(\Gamma,\Delta')$.
  For each application $(\bfsym{u_i}; R)$, $i = 1, \dots, t$, we create a copy ${\cal J}_i$ of ${\cal J}$ in which the variables $\bfsym{x}$ have been replaced by $\bfsym{u_i}$.
  We now create an instance ${\cal I}$ of \wMCSP$(\Gamma,\Delta)$ as follows:
  let $V({\cal I}) = (\bigcup_{i=1}^t V({\cal J}_i)) \cup V({\cal I'})$, $S({\cal I}) = S({\cal I'}) + M \cdot \sum_{i=1}^t S({\cal J}_i)$, and let the set of constraint applications of ${\cal I}$ consist of all applications from ${\cal I'}$ apart from those involving the relation $R$, and all applications from ${\cal J}_i$, $i = 1, \dots, t$.
  We will choose $M$ large enough, so that if $\mathcal{I'}$ is
  satisfiable, then in any optimal solution $\sigma$ to $\mathcal{I}$, 
  the restriction of $\sigma$ to the set $V(\mathcal{J}_i)$
  is forced to be an optimal solution to the instance $\mathcal{J}_i$.
  It then follows that $\sigma(\bfsym{u}_i) \in R$,
  so we can recover an optimal solution to $\mathcal{I}'$ from $\sigma$.
  The value of $M$ is chosen as follows:
  if all solutions to $\mathcal{J}$ have the same measure, we can let $M = 0$.
  Otherwise,
  let $\delta > 0$ be the minimal difference in measure between a
  sub-optimal solution, and an optimal solution to $\mathcal{J}$.
  Assume that $S({\cal I'}) = \sum_{i=1}^m \rho_i f_i(\bfsym{x_i})$, and let
  $U = \sum_{i=1}^m \rho_i$.
  Note that if $\sigma$ is any solution to the instance obtained from ${\cal I'}$ by removing all constraint applications, then
  $m(\sigma) \leq U$.
  We can then let $M = (U+1) / \delta$;
  the representation size of $M$ is linearly bounded in the size of the instance ${\cal I'}$.
  It is easy to check that if
  ${\cal I}$ is unsatisfiable, or if
  ${\sf Opt}({\cal I}) > U + M \cdot t \cdot {\sf Opt}({\cal J})$, then
  ${\cal I'}$ is unsatisfiable.
  Otherwise ${\sf Opt}({\cal I'}) = {\sf Opt}({\cal I})-M \cdot t \cdot {\sf Opt}({\cal J})$.
  \qed
\end{proof}

\subsection{Multimorphisms and submodularity}

We now turn our attention to {\em multimorphisms} and tractable
minimisation problems.
Let $D$ be a finite set.
Let $f : D^k \rightarrow D$ be a function, and let $\bfsym{x}_1, \dots, \bfsym{x}_k \in D^n$, with components $\bfsym{x}_i = (x_{i1},\dots,x_{in})$.
Then, we let $f(\bfsym{x}_1,\dots,\bfsym{x}_k)$ denote the $n$-tuple $(f(x_{11},\dots,x_{k1}),\dots,f(x_{1n},\dots,x_{kn}))$.

A (binary) \emph{multimorphism} of $\Gamma$ is a pair of functions $f, g : D^2 \rightarrow D$ such that for any $h \in \Gamma$, and matching tuples $\bfsym{x}$ and $\bfsym{y}$,
\begin{equation}
h(f(\bfsym{x},\bfsym{y}))+h(g(\bfsym{x},\bfsym{y})) \leq h(\bfsym{x})+h(\bfsym{y}).
\end{equation}
The concept of multimorphisms was introduced by Cohen \etal\cite{Cohen:etal:AI06} as an extension of the notion of \emph{polymorphisms} to the analysis of the \VCSP{} problem.

\begin{definition}[Multimorphism Function Minimisation]
Let $X$ be a finite set of triples $(D_i; f_i, g_i)$, where $D_i$ is a finite set and $f_i, g_i$ are functions mapping $D_i^2$ to $D_i$. MFM$(X)$ is a minimisation
problem with

\begin{description}
\item[Instance:] A positive integer $n$, a function $j : [n] \to [|X|]$,
and a function $h : D \to \mathbb{Z}$ where $D = \prod_{i=1}^n D_{j(i)}$. Furthermore,
\begin{align}
h(\bfsym{x}) + h(\bfsym{y}) \geq
\ &h(f_{j(1)}(x_1, y_1), f_{j(2)}(x_2, y_2), \ldots, f_{j(n)}(x_n, y_n)) \ + \notag \\
  &h(g_{j(1)}(x_1, y_1), g_{j(2)}(x_2, y_2), \ldots, g_{j(n)}(x_n, y_n))   \notag
\end{align}
for all $\bfsym{x}, \bfsym{y} \in D$. The function $h$ is given to the algorithm as an oracle, \ie{}, for any $\bfsym{x} \in D$ we can query the oracle to
obtain $h(\bfsym{x})$ in unit time.
\item[Solution:] A tuple $\bfsym{x} \in D$.
\item[Measure:] The value of $h(\bfsym{x})$.
\end{description}
\end{definition}
For a finite set $X$ we say that MFM$(X)$ is \emph{oracle-tractable} if it can be solved in time $O(n^c)$ for some constant $c$.
It is not hard to see that if $(f,g)$ is a multimorphism of $\Gamma$, and
MFM$(D; f, g)$ is oracle-tractable,
then \wMCSP$(\Gamma)$ is tractable.

We now give two examples of oracle-tractable problems.
A partial order on $D$ is called a \emph{lattice} if every pair of elements $a, b \in D$ has a greatest lower bound $a \meet b$ (meet) and a least upper bound $a \join b$ (join).
A \emph{chain} on $D$ is a lattice which is also a total order.

For $i = 1, \dots, n$, let $L_i$ be a lattice on $D_i$.
The \emph{product lattice} $L_1 \times \dots \times L_n$ is defined on the set $D_1 \times \dots \times D_n$ by extending the meet and join component-wise:
for $\bfsym{a} = (a_1,\dots,a_n)$ and $\bfsym{b} = (b_1,\dots,b_n)$,
let $\bfsym{a} \meet \bfsym{b} = (a_1 \meet b_1, \dots, a_n \meet b_n)$, and
let $\bfsym{a} \join \bfsym{b} = (a_1 \join b_1, \dots, a_n \join b_n)$.

A function $f : D^k \rightarrow \Z$ is called \emph{submodular} on the lattice $L = (D;\meet,\join)$ if
\[
f(\bfsym{a} \meet \bfsym{b})+f(\bfsym{a} \join \bfsym{b}) \leq f(\bfsym{a})+f(\bfsym{b})
\]
for all $\bfsym{a}, \bfsym{b} \in D^k$.
A set of functions $\Gamma$ is said to be submodular on $L$ if every function in $\Gamma$ is submodular on $L$.
This is equivalent to $(\meet, \join)$ being a multimorphism of $\Gamma$.
It follows from known algorithms for submodular function minimisation that MFM$(X)$ is oracle-tractable for any finite set $X$ of finite \emph{distributive lattices} (\eg{} chains)~\cite{IwataFF01,Schrijver00}.

The second example is strongly related to submodularity, but here we use a partial order that is not a lattice to define the multimorphism.
Let $D = \{0,1,2\}$, and define the functions $u, v : D^2 \rightarrow D$ by letting $u(x,y) = \min \{x,y\}$, $v(x,y) = \max \{x,y\}$ if $\{x,y\} \neq \{1,2\}$, and $u(x,y) = v(x,y) = 0$ otherwise.
We say that a function $h : D^k \rightarrow \Z$ is \emph{bisubmodular} if $h$ has the multimorphism $(u, v)$.
It is possible to minimise a $k$-ary bisubmodular function in time polynomial in $k$, provided that evaluating $h$ on a tuple is a primitive operation~\cite{McCormickF10}.

\section{A New Tractable Class}\label{sec:tractable}

In this section, we introduce a new multimorphism which ensures tractability
for \wMCSP{} (and more generally for \VCSP).

\begin{definition}
  Let $b$ and $c$ be two distinct elements in $D$.
  Let $(D;\po)$ be a partial order which relates all pairs of elements except for $b$ and $c$.
  Assume that $f, g : D^2 \rightarrow D$ are two commutative functions satisfying the following conditions:
  \begin{itemize}
  \item
    If $\{x,y\} \neq \{b,c\}$, then $f(x,y) = x \glb y$ and $g(x,y) = x \lub y$.
  \item
    If $\{x,y\} = \{b,c\}$, then $\{f(x,y),g(x,y)\} \cap \{x,y\} = \emptyset$, and $f(x,y) < g(x,y)$.

  \end{itemize}
  We call $(D;f,g)$ a \emph{1-defect chain} (over $(D;<)$), and say that $\{b,c\}$ is the \emph{defect} of $(D;f,g)$.
  If a function has the multimorphism $(f,g)$, then we also say that $(f,g)$ is a \emph{1-defect chain multimorphism}.
\end{definition}

Three types of 1-defect chains are shown in Fig.~\ref{fig:1}(a--c).
Note this is not an exhaustive list, \eg{} for $|D| > 4$,
there are 1-defect chains similar to Fig.~\ref{fig:1}(b), 
but with $f(b,c) < g(b,c) < b, c$.
When $|D| = 4$, type (b) is precisely the product lattice 
shown in Fig.~\ref{fig:1}(d).
We denote this lattice by $L_{ad}$

\begin{figure}[htbp]
\begin{center}
    \begin{tikzpicture}
      \useasboundingbox (-1.5,-1.75) rectangle (1.25,1.7);
      \node () at (0,-2) {(a)};
      \node[vertex8] (ua1) at (0,-1.2cm) [] {};
      \node[vertex8] (ua) at (0,-0.6cm) [] {} edge[edged](ua1);
      \node[] (ual) at (-0.6,-0.6cm) [] {$f(b,c)$};
      \node[vertex8] (ud) at (0,0) [] {} edge[edged](ua);
      \node[] (udl) at (-0.6,0) [] {$g(b,c)$};
      \node[vertex8] (ud1) at (0,0.6cm) [] {} edge[edged](ud);
      \node[vertex8] (ub) at (-0.7cm,1.2cm) [label=left:$b$] {} edge[edge](ud1);
      \node[vertex8] (uc) at (0.7cm,1.2) [label=right:$c$] {} edge[edge](ud1);
    \end{tikzpicture}
    \begin{tikzpicture}
      \useasboundingbox (-1.5,-1.75) rectangle (1.25,1.7);
      \node () at (0,-2) {(b)};
      \node[vertex8] (ua1) at (0,-1.5cm) [] {};
      \node[vertex8] (ua) at (0,-0.9cm) [] {} edge[edged](ua1);
      \node[] (ual) at (-0.6,-0.9cm) [] {$f(b,c)$};
      \node[vertex8] (ua2) at (0,-0.3cm) [] {} edge[edged](ua);
      \node[vertex8] (ub) at (-0.7cm,0) [label=left:$b$] {} edge[edge](ua2);
      \node[vertex8] (uc) at (0.7cm,0) [label=right:$c$] {} edge[edge](ua2);
      \node[vertex8] (ud1) at (0,0.3cm) [] {} edge[edge](ub) edge[edge](uc);
      \node[vertex8] (ud) at (0,0.9cm) [] {} edge[edged](ud1);
      \node[] (udl) at (-0.6,0.9cm) [] {$g(b,c)$};
      \node[vertex8] (ud2) at (0,1.5cm) [] {} edge[edged](ud);
    \end{tikzpicture}
    \begin{tikzpicture}
      \useasboundingbox (-1.5,-1.75) rectangle (1.25,1.7);
      \node () at (0,-2) {(c)};
      \node[vertex8] (ua1) at (0,1.2cm) [] {};
      \node[vertex8] (ua) at (0,0.6cm) [] {} edge[edged](ua1);
      \node[] (ual) at (-0.6,0.6cm) [] {$g(b,c)$};
      \node[vertex8] (ud) at (0,0) [] {} edge[edged](ua);
      \node[] (udl) at (-0.6,0) [] {$f(b,c)$};
      \node[vertex8] (ud1) at (0,-0.6cm) [] {} edge[edged](ud);
      \node[vertex8] (ub) at (-0.7cm,-1.2cm) [label=left:$b$] {} edge[edge](ud1);
      \node[vertex8] (uc) at (0.7cm,-1.2) [label=right:$c$] {} edge[edge](ud1);
    \end{tikzpicture}
    \begin{tikzpicture}
      \useasboundingbox (-1.5,-1.75) rectangle (1.25,1.7);
      \node () at (0,-2) {(d)};
      \node[vertex8] (ua) at (0,-1cm) [label=below:$a$] {};
      \node[vertex8] (ub) at (-0.7cm,0) [label=left:$b$] {} edge[edge](ua);
      \node[vertex8] (uc) at (0.7cm,0) [label=right:$c$] {} edge[edge](ua);
      \node[vertex8] (ud) at (0,1cm) [label=above:$d$] {} edge[edge](ub) edge[edge](uc);
    \end{tikzpicture}
\end{center}
\caption{Three types of 1-defect multimorphisms with defect $\{b,c\}$. (a) $f(b,c) < g(b,c) < b, c$. (b) $f(b,c) < b, c < g(b,c)$. (c) $b, c < f(b,c) < g(b,c)$. (d) The Hasse diagram of the lattice $L_{ad}$, a special case of (b).}\label{fig:1}
\end{figure}
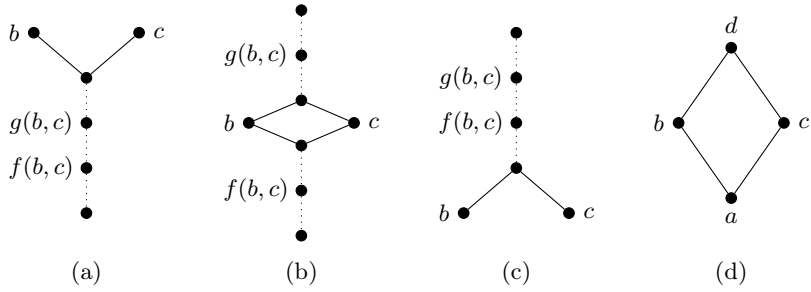

\begin{example}
\label{ex:1defect}
Let $D = \{a,b,c,d\}$, and assume that $(D;f,g)$ is a 1-defect chain, 
with defect $\{b,c\}$, and that $a = f(b,c), d = g(b,c)$.
If $a < b, c < d$, then $f$ and $g$ are the meet and join of $L_{ad}$, \cf{} Fig.~\ref{fig:1}(d). 
When $a < d < b, c$ we have the situation in Fig.~\ref{fig:1}(a),
and when $b, c < a < d$ we have the situation in Fig.~\ref{fig:1}(c).
In the two latter cases, $f$ and $g$ are given by the two following multimorphisms (rows and columns are listed in the order $a, b, c, d$, \eg{} $g_1(c,d) = c$):
\[
f_1 : 
\begin{array}{cccc}
  a & a & a & a\\
  a & b & a & d\\
  a & a & c & d\\
  a & d & d & d\\
\end{array}
\qquad
g_1 :
\begin{array}{cccc}
  a & b & c & d\\
  b & b & d & b\\
  c & d & c & c\\
  d & b & c & d\\
\end{array}
\qquad \qquad
f_2 :
\begin{array}{cccc}
  a & b & c & a\\
  b & b & a & b\\
  c & a & c & c\\
  a & b & c & d\\
\end{array}
\qquad
g_2 : 
\begin{array}{cccc}
  a & a & a & d\\
  a & b & d & d\\
  a & d & c & d\\
  d & d & d & d\\
\end{array}
\]
\end{example}

The proof of tractability for languages with 1-defect chain multimorphisms
is inspired by Krokhin and Larose's~\cite{KrokhinL08} result on maximising 
supermodular functions on Mal'tsev products of lattices.
First we will need some notation and a general lemma on oracle-tractability
of MFM problems.

For an equivalence relation $\theta$ on $D$ we use $x[\theta]$ to
denote the equivalence class containing $x \in D$.
The relation $\theta$ is a \emph{congruence} on $(D;f,g)$, if
$f(x_1,y_1)[\theta] = f(x_2,y_2)[\theta]$ and $g(x_1,y_1)[\theta] = g(x_2,y_2)[\theta]$ whenever
$x_1[\theta] = x_2[\theta]$ and $y_1[\theta] = y_2[\theta]$.
We use $D/\theta$ to denote the set $\{ x[\theta] \mid x \in D\}$ and
$f/\theta : (D/\theta)^2 \to D/\theta$ to denote the function $(x[\theta], y[\theta]) \mapsto f(x, y)[\theta]$.

\begin{lemma} \label{lem:maltsev}
Let $f, g$ be two functions that map $D^2$ to $D$. If there is a congruence relation $\theta$ on $(D; f, g)$ such that
1) MFM$(D/\theta; f/\theta, g/\theta)$ is oracle-tractable; and
2) MFM$(\{(X; f|_X, g|_X) \mid X \in D/\theta\})$ is oracle-tractable, then
MFM$(D; f, g)$ is oracle-tractable.
\end{lemma}
\begin{proof}
Let $h : D^n \to \mathbb{Z}$ be the function we want to minimise. We define a new function $h' : (D/\theta)^n \to \mathbb{Z}$ by
\[
h'(z_1, z_2, \ldots, z_n) = \min_{x_i \in z_i} h(x_1, x_2, \ldots, x_n) .
\]
It is clear that $\min_{\bfsym{z} \in (D/\theta)^n} h'(\bfsym{z}) = \min_{\bfsym{x} \in D^n} h(\bfsym{x})$.
By assumption 2 in the statement of the lemma we can compute $h'$ given $z_1, z_2, \ldots, z_n$. To simplify the notation we let $u = f/\theta$ and $v = g/\theta$. We will now prove that
$h'$ is an instance of MFM$(D/\theta; u, v)$.

Let $\bfsym{x}, \bfsym{y} \in D^k$ and choose $x'_i \in x_i[\theta]$ and $y'_i \in y_i[\theta]$ so that $h'(\bfsym{x}[\theta]) = h(\bfsym{x}')$ and $h'(\bfsym{y}[\theta]) = h(\bfsym{y}')$.
We then have
\begin{align}
h'(\bfsym{x}[\theta])+h'(\bfsym{y}[\theta]) &=    h(\bfsym{x}') + h(\bfsym{y}') \label{eq:1} \\
                                            &\geq h(f(\bfsym{x}',\bfsym{y}')) + h(g(\bfsym{x}',\bfsym{y}')) \label{eq:2} \\
                                            &\geq h'(f(\bfsym{x}',\bfsym{y}')[\theta]) + h'(g(\bfsym{x}',\bfsym{y}')[\theta]) \label{eq:3} \\
                                            &=    h'(f(\bfsym{x},\bfsym{y})[\theta])+h'(g(\bfsym{x},\bfsym{y})[\theta])) \label{eq:4} \\
                                            &=    h'(u(\bfsym{x}[\theta],\bfsym{y}[\theta]))+h'(v(\bfsym{x}[\theta],\bfsym{y}[\theta])) . \label{eq:5}
\end{align}
Here~\eqref{eq:1} follows from our choice of $\bfsym{x'}$ and $\bfsym{y'}$,
\eqref{eq:2} follows from the fact that $h$ is an instance of MFM$(D; f, g)$,
\eqref{eq:3} follows from the definition of $h'$, and finally
\eqref{eq:4} and \eqref{eq:5} follows as $\theta$ is a congruence relation of $f$ and $g$.
Hence, $h'$ is an instance of MFM$(D/\theta; u, v)$ and can be minimised in polynomial time by the first assumption in the lemma.
\qed
\end{proof}

Armed with this lemma and the oracle-tractability of submodular and bisubmodular functions described in the previous section, we can now present a new tractable class of \wMCSP-problems.

\begin{proposition}
\label{prop:1defecttract}
  If $\Gamma$ has a 1-defect chain multimorphism, then \wMCSP$(\Gamma)$ is tractable.
\end{proposition}
\begin{proof}
  Assume that $\Gamma$ has a 1-defect chain multimorphism $(f,g)$ over $(D;<)$
  with defect $\{b,c\}$. We prove that MFM$(D; f, g)$ is oracle-tractable.
  
  Assume that $b$ and $c$ are maximal elements, \ie{} $x < b, c$ for all $x \in D \setminus \{b, c\}$. In
  this case the equivalence relation $\theta$ with classes $A = D \setminus \{b, c\}$, $B = \{b\}$, $C = \{c\}$ is a 
  congruence relation of $(D; f, g)$. 
  Furthermore, MFM$(\{A,B,C\}; f/\theta, g/\theta)$ and
  MFM$(A; f|_A, g|_A)$ are oracle-tractable~\cite{McCormickF10,Schrijver00}.
  It now follows from Lemma~\ref{lem:maltsev} that MFM$(D; f, g)$ is oracle-tractable.  
  The same argument works for the case when $b$ and $c$ are minimal elements. 
  
  If $f(b, c) < g(b, c) < b, c$, but $b$ and $c$ are not maximal, then we can use the
congruence relation $\theta'$ with classes $A = \{x \mid x \leq b \text{ or } x \leq c\}$ and $B = D \setminus A$.
Here, $(\{A,B\}; f/\theta', g/\theta')$ and $(B; f|_B, g|_B)$ are chains, and $(A; f|_A, g|_A)$ is a 1-defect
chain of the previous type. 
One can show that when MFM($X$) and MFM($Y$) are both oracle-tractable,
then so is MFM$(X \cup Y)$.
Combining this with the technique used above,
we can now solve the minimisation problem. An analogous construction works
in the case when $b, c < f(b, c), g(b, c)$, using the congruence consisting of the
class $\{x \mid x \geq b \text{ or } x \geq c\}$ and its complement.
Finally, when $f(b, c) < b, c < g(b, c)$, we can use the congruence relation $\theta''$
with classes $B = \{x \mid x \leq b\}$ and $C = \{x \mid x \geq c\}$. Here, $(\{B,C\}, f/\theta'', g/\theta'')$,
$(B, f|_B, g|_B)$, and $(C, f|_C, g|_C)$ are all chains and thus the MFM problem for these triples is oracle-tractable~\cite{Schrijver00}.
\qed
\end{proof}

We now turn to prove a different property of functions with 1-defect chain
multimorphisms.
It is based on the following result for submodular functions on chains, which was derived by Queyranne \etal\cite{Queyranne:etal:98} from earlier work by Topkis~\cite{Topkis:1978}
(See also Burkard \etal\cite{Burkard:Klinz:Rudolf}).
This formulation is due to Deineko \etal\cite{Deineko:etal:JACM2008}:

\begin{lemma}
\label{lem:burkard}
  A function $f : D^k \rightarrow \Z$ is submodular on a chain $(D;\meet,\join)$
  if and only if the following holds: every binary function obtained from
  $f$ by replacing any given $k-2$ variables by any constants is submodular
  on this chain.
\end{lemma}

It is straightforward to extend this lemma to products of chains,
such as $L_{ad}$. 
Here, we outline the proof of the corresponding property for arbitrary 1-defect chains,
which will be needed in Section~\ref{sec:binary}.
We will use the following observation.

\begin{definition}
  A binary operation $f : D^2 \rightarrow D$ is called a 
  \emph{2-semilattice} if it is
  idempotent, commutative, and
  $f(f(x,y),x) = f(x,y)$ for all $x, y \in D$.
\end{definition}

\begin{proposition}
\label{prop:fg}
  Let $(D;f,g)$ be a 1-defect chain with a defect on $\{b,c\}$.
  \begin{enumerate}
  \item
    If $f(b,c) < b, c$, then $f$ is a 2-semilattice
    and $g(f(x,y),x) = x$ for $x, y \in D$.
  \item
    If $g(b,c) > b, c$, then $g$ is a 2-semilattice
    and $f(g(x,y),x) = x$ for $x, y \in D$.
  \item
    For $\bfsym{x}, \bfsym{y} \in \{b,c\}^k$, we have
    $\{f(f(\bfsym{x},\bfsym{y}),\bfsym{x}),g(f(\bfsym{x},\bfsym{y}),\bfsym{x})\} = \{f(\bfsym{x},\bfsym{y}),\bfsym{x}\}$
    and
    $\{g(g(\bfsym{x},\bfsym{y}),\bfsym{x}),f(g(\bfsym{x},\bfsym{y}),\bfsym{x})\} = \{g(\bfsym{x},\bfsym{y}),\bfsym{x}\}$.
  \end{enumerate}
\end{proposition}

\begin{proof}
  For $\{x,y\} \neq \{b,c\}$, the equalities
  $f(f(x,y),x) = f(x,y)$ and $g(f(x,y),x) = x$ 
  follow from the underlying partial order.
  Assume instead that $\{x,y\} = \{b,c\}$, and that $f(x,y) < x, y$.
  Since $\{f(x,y),x\} \neq \{b,c\}$, we have that $f(f(x,y),x)$ is the
  greatest lower bound of $f(x,y)$ and $x$, which is $f(x,y)$.
  We also have that $g(f(x,y),x)$ is the lowest upper bound of
  $f(x,y)$ and $x$, which is $x$.
  An analogous argument proves (2).

  The first equality of (3) follows from (1) if $f(b,c) < b, c$, and
  the second equality follows from (2) if $g(b,c) > b, c$.
  At least one of $f(b,c) < b, c$ and $g(b,c) > b, c$ holds.
  If both holds, there is nothing to prove, so assume that
  $f(b,c) < b, c$, but $g(b,c) < b, c$.
  We then have $g(g(x,y),x) = x$ and $f(g(x,y),x) = g(x,y)$ for $\{x,y\} = \{b,c\}$, 
  so the second equality of (3) also holds.
  The remaining case follows similarly.
\qed
\end{proof}


\begin{lemma}
\label{lem:1defectburkard}
  A function $h : D^k \rightarrow \Z$, $k \geq 2$, has the 1-defect
  chain multimorphism $(f,g)$ if and only if
  every binary function obtained from $h$ by replacing any given $k-2$ variables by any constants has the multimorphism $(f, g)$.
\end{lemma}

\begin{proof}
  Let $\{b,c\}$ be the defect of $(f,g)$.
  We prove the statement for the case $f(b,c) < b, c$.
  The other case follows analogously.

  Every function obtained from $h$ by fixing a number of variables is clearly invariant under every multimorphism of $h$.

  For the opposite direction,
  assume that $h$ does not have the multimorphism $(f, g)$.
  We want to prove that there exist
  vectors $\bfsym{x}, \bfsym{y} \in D^k$ such that
  \begin{equation}
    \label{eqn:hineq}
  h(\bfsym{x})+h(\bfsym{y}) < h(f(\bfsym{x},\bfsym{y}))+h(g(\bfsym{x},\bfsym{y})),
  \end{equation}
  with $d_H(\bfsym{x},\bfsym{y})=2$, where $d_H$ denotes the \emph{Hamming distance} on $D^k$, \ie{} the number of coordinates in which $\bfsym{x}$ and $\bfsym{y}$ differ.
  
  Assume to the contrary that the result does not hold.
  We can then choose a function $h$ of minimal arity such that
  \[
  \min \{ \text{$d_H(\bfsym{x},\bfsym{y}) \mid \bfsym{x}$ and $\bfsym{y}$ satisfy (\ref{eqn:hineq})} \} > 2.
  \]
  The arity of $h$ must in fact be equal to the least $d_H(\bfsym{x},\bfsym{y})$;
  otherwise, we could obtain a function $h'$ of strictly smaller arity
  by fixing the variables in $h$ on which $\bfsym{x}$ and $\bfsym{y}$ agree.
  This would contradict the minimality in the choice of $h$.

  We will first show that it is possible to choose $\bfsym{x}$ and $\bfsym{y}$ so that $\{x_i,y_i\} \neq \{b,c\}$ for all $i$.
  Let $k_1, k_2 \geq 1$ so that $k_1+k_2=k$, and let
  $(\bfsym{x_1}; \bfsym{x_2}), (\bfsym{y_1}; \bfsym{y_2}) \in D^k$ be two vectors with $d_H((\bfsym{x_1}; \bfsym{x_2}),(\bfsym{y_1}; \bfsym{y_2})) = k$, satisfying (\ref{eqn:hineq}).
  Now, assume that $(\bfsym{x_1};\bfsym{x_2}),(\bfsym{y_1};\bfsym{y_2}) \in \{b,c\}^k$. We then have
\[
  h(\bfsym{x_1};\bfsym{x_2})+
  h(\bfsym{y_1};\bfsym{y_2}) < 
  h(f(\bfsym{x_1},\bfsym{y_1});f(\bfsym{x_2};\bfsym{y_2}))+
  h(g(\bfsym{x_1},\bfsym{y_1});g(\bfsym{x_2},\bfsym{y_2}))
  \]
  Since both $d_H((\bfsym{x_1};\bfsym{x_2}),(\bfsym{x_1};\bfsym{y_2}))$
  and $d_H((\bfsym{y_1};\bfsym{y_2}),(\bfsym{x_1};\bfsym{y_2}))$
  are strictly less that the arity of $h$, we have by assumption
  \[
    h(\bfsym{x_1}; \bfsym{x_2})+
    h(\bfsym{x_1}; \bfsym{y_2}) \geq 
    h(\bfsym{x_1}; f(\bfsym{x_2},\bfsym{y_2}))+
    h(\bfsym{x_1}; g(\bfsym{x_2},\bfsym{y_2})), \text{ and}
  \]
  \[
    h(\bfsym{y_1}; \bfsym{y_2})+
    h(\bfsym{x_1}; \bfsym{y_2}) \geq 
    h(f(\bfsym{x_1},\bfsym{y_1}); \bfsym{y_2})+
    h(g(\bfsym{x_1},\bfsym{y_1}); \bfsym{y_2}).
  \]
  By combining these inequalities, we get
  \[
  h(\bfsym{x_1};f(\bfsym{x_2},\bfsym{y_2}))+
  h(f(\bfsym{x_1},\bfsym{y_1});\bfsym{y_2})+
  h(\bfsym{x_1};g(\bfsym{x_2},\bfsym{y_2}))+
  h(g(\bfsym{x_1},\bfsym{y_1});\bfsym{y_2})
  \]
  \[
  < h(f(\bfsym{x_1},\bfsym{y_1});f(\bfsym{x_2};\bfsym{y_2}))+
  h(\bfsym{x_1},\bfsym{y_2})+
  h(g(\bfsym{x_1},\bfsym{y_1});g(\bfsym{x_2},\bfsym{y_2}))+
  h(\bfsym{x_1},\bfsym{y_2}).
  \]
  Let 
  $\bfsym{x} = (\bfsym{x_1};f(\bfsym{x_2},\bfsym{y_2}))$,
  $\bfsym{y} = (f(\bfsym{x_1},\bfsym{y_1});\bfsym{y_2})$,
  $\bfsym{x}' = (\bfsym{x_1};g(\bfsym{x_2},\bfsym{y_2}))$, and
  $\bfsym{y}' = (g(\bfsym{x_1},\bfsym{y_1});\bfsym{y_2})$.
  By Proposition~\ref{prop:fg}(3), we have
  \[
  \{f(\bfsym{x},\bfsym{y}),g(\bfsym{x},\bfsym{y})\} = 
\{(\bfsym{x_1};\bfsym{y_2}),(f(\bfsym{x_1},\bfsym{y_1});f(\bfsym{x_2},\bfsym{y_2}))\}, \text{ and}
  \]
  \[
  \{f(\bfsym{x}',\bfsym{y}'),g(\bfsym{x}',\bfsym{y}')\} = 
\{(\bfsym{x_1};\bfsym{y_2}),(g(\bfsym{x_1},\bfsym{y_1});g(\bfsym{x_2},\bfsym{y_2}))\}.
  \]
  Hence, we can rewrite the previous inequality:
  \[
  h(\bfsym{x})+h(\bfsym{y})+h(\bfsym{x}')+h(\bfsym{y}')
  \]
  \[
  < h(f(\bfsym{x},\bfsym{y}))+h(g(\bfsym{x},\bfsym{y}))+h(f(\bfsym{x}',\bfsym{y}'))+h(g(\bfsym{x}',\bfsym{y}')).
  \]

  It follows that either the pair $\bfsym{x}$ and $\bfsym{y}$, 
  or the pair $\bfsym{x}'$ and $\bfsym{y}'$ satisfies 
  condition (\ref{eqn:hineq}).
  Furthermore, $\{x_i,y_i\} \neq \{b,c\}$ and $\{x'_i,y'_i\} \neq \{b,c\}$,
  for all $i$.

  If instead we have vectors $\bfsym{x}$ and $\bfsym{y}$
  satisfying (\ref{eqn:hineq}) such that $\{x_i,y_i\} \neq \{b,c\}$ for
  \emph{some}, but not all $i$, then we proceed as follows.
  Note that $\{x_i,y_i\} \neq \{b,c\}$ implies 
  $\{f(x_i,y_i),g(x_i,y_i)\} = \{x_i,y_i\}$.
  Without loss of generality, we may therefore assume that 
  $\bfsym{x} = (\bfsym{x_1};\bfsym{x_2}),
  \bfsym{y} = (\bfsym{y_1};\bfsym{y_2}) \in D^k$,
  with $\bfsym{x_1},\bfsym{y_1} \in D^{k_1}$ for $k_1 \geq 1$,
  are such that
  $f(\bfsym{x_1},\bfsym{y_1}) = \bfsym{x_1}$ and 
  $g(\bfsym{x_1},\bfsym{y_1}) = \bfsym{y_1}$,
  possibly by first exchanging $\bfsym{x}$ and $\bfsym{y}$.
  For these vectors, condition (\ref{eqn:hineq}) now reads:
  \[
  h(\bfsym{x_1};\bfsym{x_2})+h(\bfsym{y_1};\bfsym{y_2}) <
  h(\bfsym{x_1};f(\bfsym{x_2},\bfsym{y_2}))+
  h(\bfsym{y_1};g(\bfsym{x_2},\bfsym{y_2})).
  \]
  Due to the minimality of $h$'s arity, we must have
  \[
  h(\bfsym{y_1};\bfsym{x_2})+h(\bfsym{y_1};\bfsym{y_2}) \geq
  h(\bfsym{y_1};f(\bfsym{x_2},\bfsym{y_2}))+
  h(\bfsym{y_1};g(\bfsym{x_2},\bfsym{y_2})).
  \]
  We therefore have
  \[
  h(\bfsym{x_1};\bfsym{x_2})+h(\bfsym{y_1};f(\bfsym{x_2},\bfsym{y_2})) <
  h(\bfsym{x_1};f(\bfsym{x_2},\bfsym{y_2}))+h(\bfsym{y_1};\bfsym{x_2}).
  \]
  Let $\bfsym{x} = (\bfsym{x_1}; \bfsym{x_2})$ and
  $\bfsym{y} = (\bfsym{y_1};f(\bfsym{x_2},\bfsym{y_2}))$.
  By Proposition~\ref{prop:fg}(1), $f$ is a 2-semilattice, so we have
  $f(f(\bfsym{x_2},\bfsym{y_2}),\bfsym{x_2}) = f(\bfsym{x_2},\bfsym{y_2})$,
  and thus 
  \begin{multline*}
    (\bfsym{x_1};f(\bfsym{x_2};\bfsym{y_2})) =
    (\bfsym{x_1};f(f(\bfsym{x_2},\bfsym{y_2}),\bfsym{x_2})) \\
    = f((\bfsym{y_1};f(\bfsym{x_2},\bfsym{y_2})),(\bfsym{x_1};\bfsym{x_2})) =
    f(\bfsym{y},\bfsym{x}).
  \end{multline*}
  Furthermore, $g(f(\bfsym{x_2},\bfsym{y_2}),\bfsym{x_2}) = \bfsym{x_2}$, so
  \[
  (\bfsym{y_1};\bfsym{x_2}) = (\bfsym{y_1};g(f(\bfsym{x_2},\bfsym{y_2}),\bfsym{x_2})) = g((\bfsym{y_1};f(\bfsym{x_2},\bfsym{y_2})),(\bfsym{x_1};\bfsym{x_2})) = g(\bfsym{y},\bfsym{x}).
  \]
  We therefore conclude that
  \[
  h(\bfsym{x})+h(\bfsym{y}) < h(f(\bfsym{x},\bfsym{y}))+h(g(\bfsym{x},\bfsym{y})),
  \]
  so that condition (\ref{eqn:hineq}) holds for $\bfsym{x}$ and $\bfsym{y}$
  with $\{x_i,y_i\} \neq \{b,c\}$ for all $i$.
  From now on, we assume that $\bfsym{x}$ and $\bfsym{y}$ are chosen in this way.

  Let $D' = D \setminus \{b,c\} \cup \{B\}$.
  For each $i$, let $\phi_i : D' \rightarrow D$ be an injection which fixes $D \setminus \{b,c\}$, and sends $B$ to $b$ or $c$ in such a way that $\{x_i,y_i\} \subseteq \phi_i(D)$.
Let $(D'; f', g')$ be the chain defined by $x <' y$ if $x, y \neq B$ and $x < y$,
$x <' B$ if $x < b, c$, and $B <' y$ if $b, c < y$.
Then, $\phi_i(f'(x,y)) = f(\phi_i(x),\phi_i(y))$, and
  $\phi_i(g'(x,y)) = g(\phi_i(x),\phi_i(y))$, for all $i$.
  Let $\phi(\bfsym{z}) = (\phi_1(z_1), \dots, \phi_k(z_k))$, and
  let $\bfsym{x}', \bfsym{y}' \in (D')^k$ be such that $\phi(\bfsym{x}') = \bfsym{x}$ and $\phi(\bfsym{y}') = \bfsym{y}$.
  Define $h'(\bfsym{z}') = h(\phi(\bfsym{z}'))$. Then,
  \[
  h'(\bfsym{x}')+h'(\bfsym{y}') = h(\bfsym{x})+h(\bfsym{y}) < h(f(\bfsym{x},\bfsym{y}))+h(g(\bfsym{x},\bfsym{y}))
  \]
  \[
  = h'(f'(\bfsym{x}',\bfsym{y}'))+h'(g'(\bfsym{x}',\bfsym{y}')).
  \]
  It follows that $h'$ is not submodular on $(D',f',g')$.
  By Lemma~\ref{lem:burkard},
  there are elements $\bfsym{z}', \bfsym{w}' \in (D')^k$ with $d_H(\bfsym{z}',\bfsym{w}') = 2$ such that
  $h'(\bfsym{z}')+h'(\bfsym{w}') < h'(f'(\bfsym{z}',\bfsym{w}'))+h'(g'(\bfsym{z}',\bfsym{w}'))$. Hence,
  \[
  h(\phi(\bfsym{z}'))+h(\phi(\bfsym{w}')) = h'(\bfsym{z}')+h'(\bfsym{w}') < h'(f'(\bfsym{z}',\bfsym{w}'))+h'(g'(\bfsym{z}',\bfsym{w}'))
  \]
  \[
  = h(f(\phi(\bfsym{z}'),\phi(\bfsym{w}')))+h(g(\phi(\bfsym{z}'),\phi(\bfsym{w}'))),
  \]
  and $d_H(\phi(\bfsym{z}'),\phi(\bfsym{w}')) = 2$.
  This contradicts the original choice of $h$.  
  \qed
\end{proof}

\section{Endomorphisms, cores and constants}\label{sec:endo}

In this section, we show that under a natural condition, it is possible
to add constant unary relations to $\Gamma$ without changing the
computational complexity of the corresponding \wMCSP-problem.
Let $h : D^k \rightarrow \{0,1\}$.
A function $g : D \rightarrow D$ is called an \emph{endomorphism of $h$} if
for every $k$-tuple $(x_1, \dots, x_k) \in D^k$, it holds that
$h(x_1,\dots,x_k) = 0 \implies h(g(x_1),\dots,g(x_k)) = 0$.
The function $g$ is an endomorphism of $\Gamma$ if it is an endomorphism
of each function in $\Gamma$.
The set of all endomorphisms of $\Gamma$ is denoted by \Endo{\Gamma}.
A bijective endomorphism is called an \emph{automorphism}.
The automorphisms of $\Gamma$ form a group under composition.

\begin{definition}
\label{def:core1}
  A set of functions, $\Gamma$, is said to be a \emph{core} if all of its endomorphisms are injective.
\end{definition}

The idea is that if $\Gamma$ it not a core, then we can apply a non-injective endomorphism to every function in $\Gamma$, and obtain a polynomial-time equivalent problem on a strictly smaller domain.
We can then use results previously obtained for smaller domains~\cite{Cohen:etal:AI06,Jonsson:etal:sicomp06}.
Thus, we can restrict our attention to the case when $\Gamma$ is a core.

Jeavons \etal\cite{Jeavons:Cohen:Gyssens:Constraints1999} defined the notion of an \emph{indicator problem of order $k$} for CSPs.
We will exploit indicator problems of order 1 here,
adapted to the setting of \wMCSP.

\begin{definition}
  \label{def:indicator}
Let $\Gamma$ be a finite set of $\{0,1\}$-valued functions over $D$.
Let $X_D$ denote the set containing a variable $x_d$ for each $d \in D$, and
for $\bfsym{a} = (a_1, \dots, a_k) \in D^k$, let $\bfsym{x_a} = (x_{a_1}, \dots x_{a_k}) \in X_D^k$.
The indicator problem ${\cal IP}(\Gamma)$ is defined as the
instance of \wMCSP$(\Gamma)$ with variables $X_D$,
and sum $\sum_{f_i \in \Gamma} \sum_{\bfsym{a} \in f_i^{-1}(0)} f_i(\bfsym{x_a})$,
  where $k_i$ is the arity of the function $f_i$.
\end{definition}


Let $\iota : D \rightarrow X_D$ be the function defined by $\iota(d) = x_d$.
Theorem~3.5 in \cite{Jeavons:Cohen:Gyssens:Constraints1999} implies the following property of ${\cal IP}(\Gamma)$:

\begin{proposition}
  For any finite set of functions, $\Gamma$, the set of optimal solutions to
  ${\cal IP}(\Gamma)$ is equal to $\{ \sigma : X_D \rightarrow D \mid \sigma \circ \iota \in \Endo{\Gamma} \}$.
\end{proposition}

The proof of the following result follows the lines of similar results for related problems, such as the CSP decision problem.

\begin{proposition}
\label{prop:constants}
  Let $\Gamma$ be a core over $D$.
  Then,
  \wMCSP$(\Gamma, {\cal C}_D)$ is polynomial-time reducible to
  \wMCSP$(\Gamma)$.
\end{proposition}

\begin{proof}
  Let ${\cal J}$ be an instance of \wMCSP$(\Gamma, {\cal C}_D)$.
  The only way for ${\cal J}$ to be unsatisfiable is if it contains two
  contradicting constraint applications $(y; \{a\})$ and $(y; \{b\})$,
  with $a \neq b$.
  This is easily checked in polynomial time.

  Otherwise,
  Let $\bfsym{x}$ be a list of the variables $X_D$, and
  let $R = \pi_{\bfsym{x}} \optsol({\cal IP}(\Gamma))$.
  Now modify ${\cal J}$ to an instance ${\cal J'}$ of \wMCSP$(\Gamma, R)$ as follows.
  Add the variables in $X_D$ to $V({\cal J'})$, and
  add the constraint application $(\bfsym{x}; R)$.
  Furthermore, remove each constraint $(y; \{a\})$, and replace $y$ by $x_a$ throughout the instance.
  Let $\sigma'$ be an optimal solution to ${\cal J'}$.
  Since $\Gamma$ is a core, $g = \sigma'|_{X_D} \circ \iota$ is an automorphism of $\Gamma$, and so is its inverse, $g^{-1}$.
  Hence, 
  $\sigma = g^{-1} \circ \sigma'$ is also an optimal solution to ${\cal J'}$. From $\sigma$ we easily recover a solution to ${\cal J}$ of equal measure,
  and conversely, any solution to ${\cal J}$ can be interpreted as a solution
  to ${\cal J'}$.
  It follows that we have a reduction from \wMCSP$(\Gamma,{\cal C}_D)$ to
  \wMCSP$(\Gamma,R)$.
  By Proposition~\ref{prop:1}, we finally have a reduction from
  \wMCSP$(\Gamma,R)$ to \wMCSP$(\Gamma)$. 
  \qed
\end{proof}

For $a,b \in D$, let $e_{ab} : D \rightarrow D$ denote the function $e_{ab}(a) = b$ and $e_{ab}(x) = x$ for $x \neq a$.

\begin{lemma}
\label{lem:noend}
  If $e_{ab} \not\in \Endo{\Gamma}$, then
  $\langle \Gamma, {\cal C}_D \rangle_{fn}$ contains a unary $\{0,1\}$-valued function
  $u$ such that $u(a) = 0$ and $u(b) = 1$.
\end{lemma}

\begin{proof}
  Let $h : D^k \rightarrow \{0,1\}$ be a function in $\Gamma$, and
  $a_1, \dots, a_k \in D$ be elements such that
  $h(a_1, \dots, a_k) = 0$, but $h(e_{ab}(a_1), \dots, e_{ab}(a_k)) = 1$.

  Let ${\cal J}$ be the instance of \wMCSP$(\Gamma, {\cal C}_D)$ with
  variables $V({\cal J}) = X_D$,
  sum $S({\cal J}) = h(x_{a_1},\dots,x_{a_k})$,
  and constraint applications $(x_d; \{d\})$ for $d \neq a$.
  Then, $u = {\cal J}_{x_a}$ is a unary function in $\langle \Gamma, {\cal C}_D \rangle_{fn}$, with $u(a) = 0$ and $u(b) = 1$.
  \qed
\end{proof}



\section{A Graph of Partial Multimorphisms}\label{sec:graph}

Let $\Gamma$ be a core over $D$.
In this section, we define a graph $G = (V,E)$ which
encodes either the \NP-hardness of \wMCSP$(\Gamma, {\cal C}_D)$
or provides a multimorphism for the binary functions in $\langle \Gamma, {\cal C}_D \rangle_{fn}$.
The graph is a variation of a graph defined by Kolmogorov and \v{Z}ivn{\'y}~\cite{Zivny:1}, with changes made to accommodate for additional multimorphisms.

Let $V$ be the set of partial functions $(f,g) : D^2 \rightarrow D^2$ such that
\begin{itemize}
\item
  $f$ and $g$ are defined on a subset $\{a,b\} \subseteq D$; 
\item
  $f$ and $g$ are idempotent and commutative; and
\item
  $\{f(a,b),g(a,b)\} = \{a,b\}$ or $\{f(a,b),g(a,b)\} \cap \{a,b\} = \emptyset$.
\end{itemize}

We do allow $a = b$ in the definition of $V$, \ie{} there is precisely one vertex for each singleton in $D$.
For $a,b \in D$, we let $G[a,b]$ denote the graph induced by the set of vertices defined on $\{a,b\}$.
Let $(f_1,g_1) \in G[a_1,b_1]$ and $(f_2,g_2) \in G[a_2,b_2]$. 
There is an edge in $E$ between $(f_1,g_1)$ and $(f_2,g_2)$ if
there is a binary function $h \in \langle \Gamma, {\cal C}_D \rangle_{fn}$ such that
\begin{multline}
\label{eq:basic}
\min \{ h(a_1,a_2)+h(b_1,b_2), h(a_1,b_2)+h(b_1,a_2) \} < \\ h(f_1(a_1,b_1),f_2(a_2,b_2))+h(g_1(a_1,b_1),g_2(a_2,b_2)).
\end{multline}
The following lemma shows how $G$ can be used to construct multimorphisms of binary functions in $\langle \Gamma, {\cal C}_D \rangle_{fn}$:

\begin{lemma}
\label{lem:indset}
  Let $I \subseteq V$ be an independent set in $G$ with
  precisely one vertex $(f_{\{x,y\}}, g_{\{x,y\}})$ from each subgraph $G[x,y]$.
  Then, every binary function $h \in \langle \Gamma, {\cal C}_D \rangle_{fn}$
  has the multimorphism $(f, g)$ defined by
  $f(x,y) = f_{\{x,y\}}(x,y)$ and
  $g(x,y) = g_{\{x,y\}}(x,y)$.
\end{lemma}

\begin{proof}
  Assume to the contrary that $(f,g)$ is not a multimorphism of $h$.
  Then, there are tuples $(a_1,a_2), (b_1,b_2) \in D^2$ such that
  \begin{equation*}
    h(a_1,a_2)+h(b_1,b_2) < h(f(a_1,b_1),f(a_2,b_2))+h(g(a_1,b_1),g(a_2,b_2)).
  \end{equation*}
  But this would imply that $\{(f_{\{a_1,b_1\}},g_{\{a_1,b_1\}}),(f_{\{a_2,b_2\}},g_{\{a_2,b_2\}})\} \in E$, which is a contradiction since $I$ is an independent set.
\qed
\end{proof}

For distinct $a, b \in D$, let $\consv{ab}$ denote the vertex $(f,g) \in G[a,b]$ such that $f(a,b) = f(b,a) = a$ and $g(a,b) = g(b,a) = b$.
We say that such a vertex is \emph{conservative}.
Let $\ST$ denote the set of all conservative vertices, and let
$G' = G[\ST]$ be the subgraph of $G$ induced by $\ST$.
Let $\ST_\Gamma \subseteq V'$ be the set of vertices $\consv{xy}$ such that $\{x,y\} \in \langle \Gamma, {\cal C}_D \rangle_{w}$. 
For conservative vertices $\consv{a_1 b_1}$ and $\consv{a_2 b_2}$, condition (\ref{eq:basic}) reduces to:
\begin{equation}
\label{eq:cbasic}
  h(a_1,b_2)+h(b_1,a_2) < h(a_1,a_2)+h(b_1,b_2).
\end{equation}

For a vertex $x = (f,g)$, we let $\overline{x}$ denote the vertex $(g,f)$.
It follows immediately from $(\ref{eq:basic})$ that
$\{x,y\} \in E$ iff $\{\overline{x},\overline{y}\} \in E$.
Next, we prove a number of basic properties of the graph $G$.

\begin{lemma}
  \label{lem:niceh}
  If $\{\consv{a_1 b_1}, \consv{a_2 b_2}\} \in E$, then there exists a function $h \in \langle \Gamma, {\cal C}_D \rangle_{fn}$ such that
  $h(a_1,b_2) = h(b_1,a_2) < h(a_1,a_2) = h(b_1,b_2)$.
\end{lemma}


\begin{proof} 
  By definition of $G$, we can find $f \in \langle \Gamma, {\cal C}_D \rangle_{fn}$ such that
  \begin{equation}
    \label{eqn:a1b2}
    f(a_1,b_2) + f(b_1,a_2) < f(a_1,a_2) + f(b_1,b_2).
  \end{equation}
  Since $\Gamma$ is assumed to be a core, Lemma~\ref{lem:noend} is applicable for all choices of $a$ and $b$.
  Using the unary functions obtained from this lemma, it is possible to ensure that the inequality in (\ref{eqn:a1b2}) holds for a function $f$ with $f(a_1,b_2) = f(a_2,b_1)$.
  We will also assume that $f(a_1,a_2) \geq f(b_1,b_2)$ so that
  $\gamma = (f(a_1,a_2)-f(b_1,b_2))/2 \geq 0$. 
  Let $f_{a_1}$ and $f_{a_2}$ be unary functions such that $f_{a_1}(a_1) < f_{a_1}(b_1)$ and $f_{a_2}(a_2) < f_{a_2}(b_2)$, and
  let $\alpha = f_{a_1}(b_1)-f_{a_1}(a_1)$ and $\beta = f_{a_2}(b_2)-f_{a_2}(a_2)$, and note that $\alpha, \beta > 0$.  
  Now, define
  \[
  h(x,y) = f(x,y)+\gamma \left( \alpha^{-1} f_{a_1}(x) + \beta^{-1} f_{a_2}(y) \right).
  \]
  The function $h$ satisfies the inequality $h(a_1,b_2)+h(b_1,a_2) < h(a_1,a_2)+h(b_1,b_2)$,
  and furthermore,
  \begin{multline*}
    h(a_1,a_2)-h(b_1,b_2) = f(a_1,a_2)-f(b_1,b_2) + \\ 
    + \gamma \left( \frac{f_{a_1}(a_1)-f_{a_1}(b_1)}{\alpha} + \frac{f_{a_2}(a_2)-f_{a_2}(b_2)}{\beta} \right ) = \\
  = f(a_1,a_2)-f(b_1,b_2) + \gamma (-\alpha/\alpha-\beta/\beta) = 0,
  \end{multline*}
  and
  \begin{multline*}
    h(a_1,b_2)-g(b_1,a_2) = f(a_1,b_2)-f(b_1,a_2) + \\
    + \gamma \left( \frac{f_{a_1}(a_1)-f_{a_1}(b_1)}{\alpha} + \frac{f_{a_2}(b_2)-f_{a_2}(a_2)}{\beta} \right ) = \\
    = 0 + \gamma (-\alpha/\alpha + \beta/\beta) = 0.
  \end{multline*}
  The lemma follows.
  \qed
\end{proof}

\begin{lemma}
\label{gprop}
\begin{enumerate}
\item\label{item:xyz}
  Let $x$, $y$, and $z$ be conservative vertices, with $\{x,y\}$, $\{y,z\} \in E$, and assume that $y \in \ST_\Gamma$.
  Then, $\{x,\overline{z}\} \in E$.
\item\label{item:paths}
  For $n \geq 2$, let $(x_1, \ldots, x_n)$ be a path of conservative vertices in $G$, with $x_2, \dots, x_{n-1} \in \ST_\Gamma$.
  If $n$ is even, then $\{x_1, x_n\} \in E$, otherwise $\{x_1, \overline{x_n}\} \in E$.
\item\label{item:loop}
  For $n \geq 3$, let $(x_1, \ldots, x_n, x_1)$ be an odd cycle of conservative vertices in $G$, with $x_2, \dots, x_n \in \ST_\Gamma$.
  Then, there is a loop on $x_1$.
\item\label{item:prop1}
  If $\{\consv{a_1 b_1}, \consv{a_2 b_2}\} \in E$, then for each element $x \neq a_2, b_2$, either $\{\consv{a_1 b_1}, \consv{a_2 x}\} \in E$ or $\{\consv{a_1 b_1}, \consv{x b_2}\} \in E$.
\item\label{item:trans2}
  If $\{\consv{xy},\consv{yx}\}, \{\consv{yz},\consv{zy}\} \in E$ and $\{\consv{xy}, \consv{yz}\} \not\in E$, 
  then $\{\consv{xy}, \consv{zx}\}, \{\consv{yz}, \consv{zx}\} \in E$.
\item\label{item:abac}
  If there is a loop on $\consv{xz}$, but $\consv{xy}$ and $\consv{yz}$ are loop-free, then $\{\consv{xy},\consv{yz}\} \in E$.
\end{enumerate}
\end{lemma}

\begin{proof}
  Properties (1)--(3) are minor variations of Lemma 11$(b)$ and $(e)$ in~\cite{Zivny:1}.
  We include the proofs here for completeness.

  (\ref{item:xyz})
  Let $x = \consv{a_1 b_1}$, $y = \consv{a_2 b_2}$, and $z = \consv{a_3 b_3}$.
  By Lemma \ref{lem:niceh}, we have $h_1, h_2 \in \langle \Gamma, {\cal C}_D \rangle_{fn}$ such that
  $\alpha_1 = h_1(a_1,b_2)=h_1(b_1,a_2)<h_1(a_1,a_2)=h_1(b_1,b_2)=\beta_1$ and
  $\alpha_2 = h_2(a_2,b_3)=h_2(b_2,a_3)<h_2(a_2,a_3)=h_2(b_2,b_3)=\beta_2$.
  Let $h'(u_1,u_3)=\min_{u_2 \in \{a_2,b_2\}} h_1(u_1,u_2)+h_2(u_2,u_3)$,
  which is in $\langle \Gamma, {\cal C}_D \rangle_{fn}$ since $y \in \ST$.
  Now,
  $h'(a_1,b_3)+h'(b_1,a_3)=\min_{u_2,v_2 \in \{a_2,b_2\}} h_1(a_1,u_2)+h_1(b_1,v_2)+h_2(u_2,b_3)+h_2(v_2,a_3)=2 \min\{\alpha_1+\beta_2,\alpha_2+\beta_1\}$.
  We also have
  $h'(a_1,a_3)+h'(b_1,b_3)=\min_{u_2,v_2 \in \{a_2,b_2\}} h_1(a_1,u_2)+h_1(b_1,v_2)+h_2(u_2,a_3)+h_2(v_2,b_3)=2 (\alpha_1+\beta_1)$.
  It follows that $h'(a_1,b_3)+h'(b_1,a_3) < h'(a_1,a_3)+h'(b_1,b_3)$,
  so $\{x, \overline{z}\} \in E$.

  (\ref{item:paths}) and (\ref{item:loop})
  These two properties follow by repeated application of (\ref{item:xyz}),
  keeping in mind that $\{x,y\} \in E$ iff $\{\overline{x},\overline{y}\} \in E$.


  (\ref{item:prop1})
  By definition there exists a function $h \in \langle \Gamma, {\cal C}_D \rangle_{fn}$
  such that $h(a_1,a_2)+h(b_1,b_2) > h(a_1,b_2)+h(b_1,a_2)$.
  If $h(a_1,a_2)+h(b_1,x) > h(a_1,x)+h(b_1,a_2)$, then we are in the first case.
  Otherwise, $h(a_1,a_2)+h(b_1,x) \leq h(a_1,x)+h(b_1,a_2)$,
  so $h(a_1,x)+h(b_1,b_2)=h(a_1,a_2)+h(b_1,b_2)+(h(a_1,x)-h(a_1,a_2)) >
  h(a_1,b_2)+h(b_1,a_2)+h(a_1,x)-h(a_1,a_2) \geq
  h(a_1,b_2)+(h(a_1,a_2)+h(b_1,x))-h(a_1,a_2)$, which shows that we are in the second case.

  (\ref{item:trans2})
  By (\ref{item:prop1}), $\{\consv{xy},\consv{yx}\} \in E$ implies
  $\{\consv{xy},\consv{yz}\} \in E$ or $\{\consv{xy},\consv{zx}\} \in E$.
  In the first case, we are done, so we assume that the latter holds.
  Again by (\ref{item:prop1}), $\{\consv{yz},\consv{zy}\} \in E$ implies
  $\{\consv{yz},\consv{zx}\} \in E$ or $\{\consv{yz},\consv{xy}\} \in E$.
  In the latter case, we are done, hence it follows that if
  $\{\consv{yz},\consv{xy}\} \not\in E$, then we have both
  $\{\consv{xy},\consv{zx}\}$ and $\{\consv{yz},\consv{zx}\}$ in $E$.

  (\ref{item:abac})
  By (\ref{item:prop1}), $\{\consv{xz}, \consv{xz}\} \in E$ implies
  $\{\consv{xz}, \consv{xy}\} \in E$ or $\{\consv{xz}, \consv{yz}\} \in E$.
  In the first case, this in turn implies either
  $\{\consv{xy}, \consv{xy}\} \in E$ or $\{\consv{xy}, \consv{yz}\} \in E$.
  In the second case, it implies either
  $\{\consv{yz}, \consv{xy}\} \in E$ or $\{\consv{yz}, \consv{yz}\} \in E$.
  Hence, if both $\consv{xy}$ and $\consv{yz}$ are loop-free, then
  $\{\consv{xy},\consv{yz}\} \in E$.
\qed
\end{proof}

\section{Classification for $|D| = 4$}\label{sec:binary}

We are now ready to derive a classification of the computational
complexity of \wMCSP{} over a four-element domain. From here on, we assume that $D$ is the domain $\{a,b,c,d\}$.
First, we prove a result which describes the structure of the
unary functions in $\langle \Gamma, {\cal C} \rangle_{fn}$,
when $\Gamma$ is a core.

%
%
Let $\Sigma = \{ \{x,y\} \subseteq D \mid x \neq y \}$,
 $\Sigma_{ad} = \Sigma \setminus \{\{b,c\}\}$,
 $\Sigma_0 = \Sigma \setminus \{ \{b,c\}, \{a,d\}\}$, and
let $\Sigma_\Gamma = \langle \Gamma, {\cal C}_D \rangle_{w} \cap \Sigma$.
For distinct $x, y \in D$, let $u_{xy}(z) = 0$ if $z \in \{x,y\}$, and $u_{xy}(z) = 1$ otherwise.

\begin{proposition}
\label{prop:corecases}
  Let $\Gamma$ be a core over $\{a,b,c,d\}$ and assume that $\{b,c\} \not\in \Sigma_\Gamma$.
  Then,
  $\Sigma_0 \subseteq \Sigma_\Gamma$ and
  for all unary functions $u  \in \langle \Gamma, {\cal C}_D \rangle_{fn}$,
  we have $u(a)+u(d) \leq u(b)+u(c)$.
  If $\Sigma_0 = \Sigma_\Gamma$, then 
  $u(a)+u(d) = u(b)+u(c)$.
\end{proposition}

\begin{proof}
  Let ${\cal U}$ be the set of unary functions in $\langle \Gamma, {\cal C}_D \rangle_{fn}$.
  In order to simplify notation we will denote a unary function $u$ by the vector $(u(a),u(b),u(c),u(d))$.
  To exclude the functions $e_{ba}$, $e_{ca}$, $e_{bd}$, and $e_{cd}$ from the endomorphisms of $\Gamma$, Lemma~\ref{lem:noend} states that ${\cal U}$ must contain certain unary $\{0,1\}$-valued functions.
  The following table
  lists the possibilities, provided that $\{b,c\} \not\in \Sigma_\Gamma$, so in particular $u_{bc} = (1,0,0,1) \not\in {\cal U}$.
    \begin{center}
      \renewcommand{\arraystretch}{1.25}
      \begin{tabular}{ c c c c }
	$e_{ba}$ & $e_{ca}$ & $e_{bd}$ & $e_{cd}$ \\
	\hline
	(1,0,0,0) & (1,0,0,0) & (0,0,0,1) & (0,0,0,1) \\
	(1,0,1,0) & (1,1,0,0) & (0,0,1,1) & (0,1,0,1) \\
	(1,0,1,1) & (1,1,0,1) & (1,0,1,1) & (1,1,0,1) \\
      \end{tabular}
    \end{center}

  For each of the four functions $e_{xy}$, it is necessary that at least one of the three functions in the corresponding column is in ${\cal U}$.
  First assume that $(1,0,0,0) \in {\cal U}$.
  We note that $(1,0,0,0)+(0,0,0,1) = u_{bc}$, so we conclude that $(0,0,0,1) \not\in {\cal U}$.
  Since $e_{cd}$ is not an endomorphism of $\Gamma$, we must therefore
  either have $(0,1,0,1)$ or $(1,1,0,1)$ in ${\cal U}$.
  In the former case, we can add $(1,0,0,0)$ to obtain $(1,1,0,1)$,
  so we know that $(1,1,0,1) \in {\cal U}$.
  By a similar argument, considering the function $e_{bd}$,
  we conclude that $(1,0,1,1) \in {\cal U}$.
  Since $(1,0,1,1)+(1,1,0,1) = 1 + u_{bc}$, we have reached a contradiction.
  A similar argument shows that $(0,0,0,1) \not\in {\cal U}$.

  Assume instead that $(1,0,0,0), (0,0,0,1) \not\in {\cal U}$,
  $(1,0,1,1) \in {\cal U}$.
  As noted above, we must have $(1,1,0,1) \not\in {\cal U}$, and consequently $(1,1,0,0), (0,1,0,1) \in {\cal U}$.
  But $(1,1,0,0)+(0,1,0,1)+2\cdot (1,0,1,1) = 2+u_{bc}$ so again we have a contradiction.
  Thus, the only possibility is that ${\cal U}_0 := \{u_{bd}, u_{cd}, u_{ab}, u_{ac}\} \subseteq {\cal U}$, so $\Sigma_0 \subseteq \Sigma_\Gamma$.

  It is not hard to see that one can write every unary function $u$ such that $u(a)+u(d) = u(b)+u(c)$ as a linear combination of functions from ${\cal U}_0$ with non-negative coefficients.
  We show that if $v \in \langle \Gamma, {\cal C}_D \rangle_{fn}$ is a unary function in such that $v(a)+v(d) < v(b)+v(c)$, then $\{a,d\} \in \Sigma_\Gamma$.
  The full statement follows similarly.

  Let $\delta = (v(b)+v(c)-v(a)-v(d))/2 > 0$, and let $M = \max_{x \in D} v(x)$.
  Define $v'(x) = M-v(x)$ if $x = b, c$, and $v'(x) = M-v(x)+\delta$ otherwise.
  Then, $v'(a)+v'(d) = v'(b)+v'(c)$, and 
  $M+\delta u_{ad} = v'+v$ can be written as a linear combination of functions
  from ${\cal U}_0 \cup \{v\}$ with non-negative coefficients.
  Hence $M+\delta u_{ad} \in \langle \Gamma, {\cal C}_D \rangle_{fn}$,
  and $\{a,d\} \in \Sigma_\Gamma$.
  \qed
\end{proof}

We need the following two propositions in order to prove 
Theorem~\ref{thm:fourclass}. 
Their proofs are deferred to the next section.

\begin{proposition}
\label{prop:6}
  Assume that $\Sigma_0 \subseteq \Sigma_\Gamma$, and that $G'$ is bipartite.
  Then, the set of binary functions in  $\langle \Gamma, {\cal C}_D \rangle_{fn}$ is submodular on a chain.
\end{proposition}

\begin{proposition}
\label{prop:fullprop}
  Assume that $\Sigma_0 \subseteq \Sigma_\Gamma$, that $G'$ is not bipartite,
  but that $G[\ST_\Gamma]$ is.
  Then, the set of binary functions in  $\langle \Gamma, {\cal C}_D \rangle_{fn}$ has a 1-defect chain multimorphism.
\end{proposition}



\begin{theorem}
\label{thm:fourclass}
  Let $\Gamma$ be a core over $D$ with $D = {a,b,c,d}$.
  If $\Gamma$ is submodular on a chain, 
  or if $\Gamma$ has a 1-defect chain multimorphism, 
  then \wMCSP$(\Gamma)$ is tractable.
  Otherwise, it is \NP-hard.
\end{theorem}

\begin{proof}
  Assume that $G[\ST_\Gamma]$ has a loop on a vertex $\consv{xy}$.
  It then follows from Lemma~\ref{lem:niceh} that there is a function 
  $h \in \langle \Gamma, {\cal C}_D \rangle_{fn}$
  such that $h(x,y)=h(y,x)<h(x,x)=h(y,y)$, 
  and $\{x,y\} \in \langle \Gamma, {\cal C}_D \rangle_{w}$. 
  By Proposition~5.1 in~\cite{Cohen:etal:AI06}, \wMCSP$(\Gamma, {\cal C}_D)$ is \NP-hard.
  By Proposition~\ref{prop:constants}, \wMCSP$(\Gamma, {\cal C}_D)$ reduces to \wMCSP$(\Gamma)$.
  Hence, the latter is also \NP-hard.

  If instead $G[\ST_\Gamma]$ is loop-free, then it is bipartite, by Lemma~\ref{gprop}(\ref{item:loop}).
  We may assume that $\Sigma_0 \subseteq \Sigma_\Gamma$:
  this is trivial if $\Sigma_\Gamma = \Sigma$. If $\Sigma_\Gamma$ is
  strictly contained in $\Sigma$, then
  up to an automorphism we may assume that $\{b,c\} \not\in \Sigma_\Gamma$,
  and the inclusion follows by Proposition~\ref{prop:corecases}.
  For a $k$-ary function $h \in \Gamma$, let $\Phi(h)$ be the set of binary which can be obtained from $h$ by fixing at least $k-2$ variables, and let $\Gamma'$ be the union of $\Phi(h)$ over all $h \in \Gamma$.

  Now, if $G'$ is bipartite, then by Proposition~\ref{prop:6}, the set of binary functions in $\langle \Gamma, {\cal C}_D \rangle_{fn}$ is submodular on a chain.
  Since this set contains $\Gamma'$, we may conclude, by Lemma~\ref{lem:burkard},
  that $\Gamma$ is submodular on this chain as well.
  It follows that \wMCSP$(\Gamma)$ is tractable~\cite{IwataFF01,Schrijver00}.

  Otherwise, $G'$ is not bipartite, and by Proposition~\ref{prop:fullprop}, the set of binary functions in $\langle \Gamma, {\cal C}_D \rangle_{fn}$ have a 1-defect chain multimorphism.
  Since this set contains $\Gamma'$, we may conclude, by Lemma~\ref{lem:1defectburkard} this time, that $\Gamma$ has a 1-defect chain multimorphism.
  It now follows from Proposition~\ref{prop:1defecttract} that \wMCSP$(\Gamma)$ is tractable.
\qed
\end{proof}

\section{Proofs of Propositions~\ref{prop:6} and~\ref{prop:fullprop}}\label{sec:full}

\begin{lemma}
\label{lem:noniso}
  If $\Sigma_0 \subseteq \Sigma_\Gamma$, and
  $x \in V'$ is not isolated in $G'$, then
  $\{x,\overline{x}\} \in E$.
\end{lemma}

\begin{proof}
  By assumption, there is an edge $\{x,\consv{yz}\} \in E$.
  If $\{y,z\} \neq \{b,c\}, \{a,d\}$, then $\consv{yz} \in \ST_\Gamma$ since $\Sigma_0 \subseteq \Sigma_\Gamma$.
  If instead $\{x,\consv{bc}\} \in E$, then it follows from Lemma~\ref{gprop}(\ref{item:prop1}) that either $\{x,\consv{ba}\} \in E$ or $\{x,\consv{ac}\} \in E$, and $\consv{ba}, \consv{ac} \in \ST_\Gamma$ due to $\Sigma_0 \subseteq \Sigma_\Gamma$.
  In either case, $\{x,\overline{x}\} \in E$ follows from Lemma~\ref{gprop}(\ref{item:xyz}).
\qed
\end{proof}


%

For an independent set $I$ in $G'$, let $R_I$ denote the binary relation on $D$ defined by $(x,y) \in R_I$ iff $\consv{xy} \in I$.


\begin{proof}[Proposition~\ref{prop:6}]
  Let $\{I,J\}$ be a 2-colouring of the subgraph of $G'$ induced by the
  non-isolated vertices.
  We first show that $R_I$ is a partial order on $D$.
  Let $(x,y), (y,z) \in R_I$.
  Then, $\consv{xy}$ and $\consv{yz}$ have the same colour in $I$, and
  it follows that $\{\consv{xy},\consv{yz}\} \not\in E$.
  Hence, by Lemma~\ref{gprop}(\ref{item:trans2}), we have
  $\{\consv{xy},\consv{zx}\}, \{\consv{yz},\consv{zx}\} \in E$.
  By Lemma~\ref{lem:noniso}, $\{\consv{zx},\consv{xz}\} \in E$, so $\consv{xz} \in I$ and $(x,z) \in R_I$.
  Now, let $(D;<)$ be a linear extension of $R_I$, and let $I' \supseteq I$ be the corresponding subset of $\ST$.
  The set $I'$ is independent since $I$ is independent and $I' \setminus I$ is a set of isolated vertices in $G'$.
  Since there are no edges from $\ST$ to the singleton vertices in $G$, we can add all of these to $I'$ as well.
  Thus, by Lemma~\ref{lem:indset},
  every binary function in $\langle \Gamma, {\cal C}_D \rangle_{fn}$ is submodular on the
  chain $(D;\glb,\lub)$, where $\glb$ and $\lub$ are defined with respect
  to the total order $(D;<)$.
\qed
\end{proof}

%

In the following, we will let $(f,g)$ denote the vertex in $G$ given by
$f(b,c)=f(c,b)=a$ and $g(b,c)=g(c,b)=d$.

\begin{lemma}
\label{lem:fgz}
  Assume that $\Sigma_\Gamma \subseteq \Sigma_{ad}$ and that there is an edge
  $\{(f,g),z\} \in E$, $z \in V'$.
  Then,
  $\{\consv{ab},z\} \in E$ or $\{\consv{ac},z\} \in E$, and
  $\{\consv{bd},z\} \in E$ or $\{\consv{cd},z\} \in E$.
\end{lemma}

\begin{proof}
  Let $z = \consv{xy}$. 
  By definition, there exists a function $h \in \langle \Gamma, {\cal C}_D \rangle_{fn}$ such that
  $\min \{h(b,x)+h(c,y), h(c,x)+h(b,y)\} < h(a,x)+h(d,y)$.
  If $h(b,x)+h(c,y) < h(a,x)+h(d,y)$, then
  $h(a,x)+h(b,y) > (h(b,x)+h(c,y)-h(d,y))+h(b,y) \geq h(b,x)+h(a,y)+h(d,y)-h(d,y)$ since $h(b,y)+h(c,y) \geq h(a,y)+h(d,y)$
  by Proposition~\ref{prop:corecases}.
  Thus, $\{\consv{ab},\consv{xy}\} \in E$.
  If $h(c,x)+h(b,y) < h(a,x)+h(d,y)$, then we obtain
  $\{\consv{ac},\consv{xy}\} \in E$ following a similar argument,
  and the remaining two cases can be deduced in the same way.
\end{proof}

\begin{lemma}
\label{lem:u0}
  If $\Sigma_\Gamma \subseteq \Sigma_0$,
  and there is a loop on $\consv{bc}$ or $\consv{ad}$, then there is a loop on at least one of the vertices
  $\consv{ab}$, $\consv{ac}$, $\consv{bd}$, $\consv{cd}$.
\end{lemma}

\begin{proof}
  Assume, without loss of generality, that
  there exists an $h \in \langle \Gamma, {\cal C}_D \rangle_{fn}$ such that $h(b,b)+h(c,c) > h(b,c)+h(c,b)$.
  By Proposition~\ref{prop:corecases}, $\Sigma_\Gamma \subseteq \Sigma_0$ implies the relations $h(b,b)+h(c,c) = h(a,a)+h(d,d)$, 
  $h(b,b)+h(c,b) = h(a,b)+h(d,b)$, $h(b,c)+h(c,c) = h(a,c)+h(d,c)$, $h(b,b)+h(b,c) = h(b,a)+h(b,d)$, and $h(c,b)+h(c,c) = h(c,a)+h(c,d)$.
  It follows that
  $(h(a,a)+h(b,b))+(h(a,a)+h(c,c))+(h(b,b)+h(d,d))+(h(c,c)+h(d,d)) >
  2(h(b,b)+h(c,c)+h(b,c)+h(c,b)) = (h(a,b)+h(b,a))+(h(a,c)+h(c,a))+(h(b,d)+h(d,b))+(h(c,d)+h(d,c))$, which implies that the inequality $h(x,x)+h(y,y) > h(x,y)+h(y,x)$ holds in at least one of the cases $\{x,y\} = \{a,b\}, \{a,c\}, \{b,d\}, \{c,d\}$.
\qed
\end{proof}


\begin{proof}[Proposition~\ref{prop:fullprop}]
  We follow a strategy similar to that of Proposition~\ref{prop:6}.
  However, instead of using $G'$ we now consider the graph
  $G[\ST_{ad} \cup \{(f,g),(g,f)\}]$, where
  $\ST_{ad} = \ST \setminus \{\consv{bc},\consv{cb}\}$.
  First, we show that $G[\ST_{ad}]$ is bipartite.
  If $\Sigma_\Gamma = \Sigma_{ad}$, then $G[\ST_{ad}] = G[\ST_\Gamma]$ is bipartite by assumption.
  Otherwise, $\Sigma_\Gamma = \Sigma_0$.
  Since $G[\ST_\Gamma] = G[\ST_0]$ is loop-free, we know from Lemma~\ref{lem:u0} that there is no loop on $\consv{bc}$, nor on $\consv{ad}$.
  Thus, by Lemma~\ref{gprop}(\ref{item:loop}), $G[\ST_{ad}]$ is bipartite.

  Assume for the moment that the following holds:
  \begin{equation} \label{eqn:property}
    \text{For $y \in D \setminus \{b,c\}$, there is an odd path in $G[\ST_{ad}]$ from $\consv{by}$ to $\consv{yc}$}.
  \end{equation}
  
  Let $\{I,J\}$ be a 2-colouring of the subgraph of $G[\ST_{ad}]$ induced by the non-isolated vertices.
  We claim that $R_I$ is a partial order on $D$.
  Let $(x,y), (y,z) \in R_I$ and observe that
  (\ref{eqn:property}) implies $\{x,z\} \neq \{b,c\}$.
  As in the proof of Proposition~\ref{prop:6}, we can argue that $\consv{xz}$ is connected by even paths to both $\consv{xy}$ and $\consv{yz}$.
  Since $\{x,z\} \neq \{b,c\}$, it follows that $(x,z) \in I$.
 Now take a transitive extension of $R_I$ which orders all pairs of elements except for $b$ and $c$, and
  let $I' \supseteq I$ be the corresponding subset of $\ST_{ad}$.
  We can assume (possibly by swapping $I$ and $J$) that $\consv{ad} \in I'$.

  Next we show that $I' \cup \{(f,g)\}$ is independent.
  This will ensure that $f(b,c) = a < d = g(b,c)$ holds
  in the constructed multimorphism.
  If $(f,g)$ is not connected to any vertex in $\ST_{ad}$, then
  $I' \cup \{(f,g)\}$ is trivially independent.
  Otherwise, by Lemma~\ref{lem:fgz}, (\ref{eqn:property}),
  and Lemma~\ref{lem:noniso}, we can show that from any $z \in \ST_{ad}$
  such that $\{(f,g),z\} \in E$, there are odd paths in $G[\ST_{ad}]$
  to each vertex in the set
  $S = \{\consv{ab}, \consv{ac}, \consv{bd}, \consv{cd}\}$.
  Since $G[\ST_{ad}]$ is bipartite, it follows that $\{\consv{ab},\consv{bd}\} \not\in E$, so $\{\consv{ab},\consv{da}\} \in E$ by Lemma~\ref{gprop}(\ref{item:trans2}).
  Hence, $I' = I = S \cup \{\consv{ad}\}$, and $z \not\in I'$.

  It remains to verify that $I' \cup \{(f,g)\}$ together with the singleton vertices in $G$ also form an independent set,
  \ie{} that there is no edge between a singleton and $(f,g)$.
  But by condition (\ref{eq:basic}) this is equivalent to saying that each row and column of every binary function in $\langle \Gamma, {\cal C}_D \rangle_{fn}$ is submodular on $L_{ad}$, which follows from Proposition~\ref{prop:corecases}.
  By Lemma~\ref{lem:indset}, every binary function in $\langle \Gamma, {\cal C}_D \rangle_{fn}$ has the 1-defect chain multimorphism corresponding to $I' \cup \{(f,g)\}$.
  \medskip

  Finally, we prove property (\ref{eqn:property}).
  If $\Sigma_\Gamma = \Sigma_{ad}$, then
  by Lemma~\ref{gprop}(\ref{item:loop}), and the fact that $G'$ contains an odd cycle, we have a loop on $\consv{bc}$.
  Since $\consv{by}$ and $\consv{yc}$ are loop-free for $y \in D \setminus \{b,c\}$, we have $\{\consv{by},\consv{yc}\} \in E$ by Lemma~\ref{gprop}(\ref{item:abac}).
  Otherwise, $\Sigma_\Gamma = \Sigma_0$.
  We argued above that $G'$ does not contain any loop in this case.
  Thus, by Lemma~\ref{gprop}(\ref{item:loop}), every odd cycle $C$ in $G'$ must intersect both $\{\consv{bc},\consv{cb}\}$ and $\{\consv{ad},\consv{da}\}$.
  Now, by repeatedly applying Lemma~\ref{gprop}(\ref{item:paths}) to $C$, we obtain a triangle on a subset of $\{\consv{bc},\consv{cb},\consv{ad},\consv{da}\}$.
  By Lemma~\ref{lem:noniso}, we can conclude that $G'$ in fact contains the complete graph on these four vertices.
  In particular, we have both $\{\consv{ad},\consv{bc}\} \in E$ and $\{\consv{da},\consv{bc}\} \in E$.
  By Lemma~\ref{gprop}(\ref{item:prop1}), we therefore have either 
  $\{\consv{ad},\consv{ba}\} \in E$ or $\{\consv{ad},\consv{ac}\} \in E$, 
  and furthermore, either
  $\{\consv{da},\consv{ba}\} \in E$ or $\{\consv{da},\consv{ac}\} \in E$.
  Since there is no loop on $\consv{ad}$, we conclude that either
  the path $(\consv{ba},\consv{ad},\consv{da},\consv{ac})$ or
  the path $(\consv{ba},\consv{da},\consv{ad},\consv{ac})$ is in $G[\ST_{ad}]$.
  In the same way, we find an odd path from $\consv{bd}$ to $\consv{dc}$.
\qed  
\end{proof}

\section{Discussion}\label{sec:disc}

We have presented a complete complexity classification for \wMCSP{}
over a four-element domain.
More importantly, we have compiled a powerful set of tools which
will allow further systematic study of this problem.
In particular, we have shown that it is possible to add
(crisp) constants to an arbitrary core, without changing
the complexity of the problem.
This result holds in the more general case of \VCSP{} as well
(although this requires a slightly different definition of
endomorphisms),
thus answering Question~4 in \v{Z}ivn{\'y}~\cite{ZivnyPhd}.
We have also demonstrated that the techniques used by
Krokhin and Larose~\cite{KrokhinL08} for lattices can be used
effectively in the context of arbitrary algebras as well, and in doing so,
we have given the first example of an instance where submodularity 
does not suffice as an origin of tractability for \wMCSP. 
We hope that this insight will
inspire an interest in the search for more tractable cases which are
not explained by submodularity.
Finally, we have shown that graph representations such as the one
defined by Kolmogorov and \v{Z}ivn{\'y}~\cite{Zivny:1} can be used
to great effect, even in non-conservative settings.

The curious readers may ask themselves several questions at this
point, and the following one is particularly important: do 1-defect chain
multimorphisms define
genuinely new tractable classes? There is still a possibility that 
the tractability can be explained in terms of submodularity. 
We answer this question negatively with the following example.

\begin{example}
Consider the language 
$\Gamma = \{u_{bd}, u_{cd}, u_{ab}, u_{ac}, h\}$
  where $h: D^2 \rightarrow \{0,1\}$ is defined such that $h(x,y)=1$ if and only
if $x=c$ or $y=b$.
Then, $\Gamma$ is a core on $\{a,b,c,d\}$
but it is not submodular on any lattice.
However, $\Gamma$ has the 1-defect chain multimorphisms
$(f_1,g_1)$ and $(f_2,g_2)$ from Example~\ref{ex:1defect}.
\end{example}

A related question is why bisubmodularity does not appear in
the classification of {\sc Min CSP} over domains of size
three~\cite{Jonsson:etal:sicomp06}. 
The reason is that for any cost function $h : \{0,1,2\}^k \rightarrow \{0,1\}$ which
is bisubmodular, the tuple $(0, 0, \dots, 0)$ minimises $h$.
It follows that any $\{0,1\}$ constraint language over three elements 
which is bisubmodular is not a core.

There are several ways of extending this work, and one obvious way
is to study \VCSP\ instead of \wMCSP{}.
It is known that the 
\emph{fractional polymorphisms} of the constraint language, 
introduced by Cohen \etal\cite{Cohen:Cooper:Jeavons:CP06}, 
characterise the complexity of this problem
(see also~\cite{CohenCJZ10}).
Multimorphisms are a special case of such fractional polymorphisms.
As in the case of \wMCSP{}, it is currently not known
if submodularity over every finite lattice implies tractability for
\VCSP.
Distributive lattices imply tractability, and certain
constructions on lattices preserve tractability (homomorphic images
and Mal'tsev products)~\cite{KrokhinL08}. 
Furthermore, the five
element modular non-distributive lattice (also known as the diamond)
implies tractability for \emph{unweighted} \VCSP~\cite{Kuivinen09}. 
Finally, it is known that submodularity over finite modular lattices 
implies containment in {\bf NP} $\cap$ {\bf coNP}~\cite{Kuivinen09}.
It is thus clear that in order to approach further classification of
either \wMCSP{} or \VCSP{}, it will be necessary to study the
complexity of minimising submodular cost functions over new finite lattices.

As a last note, we mention that it seems to be possible to adapt
Proposition~\ref{prop:6} to the classification in~\cite{Zivny:1} of \VCSP{} 
for conservative finite-valued languages.
This would yield a simpler description of those tractable cases.

\bibliographystyle{abbrv}



\end{document}